\documentclass[submission,copyright,creativecommons]{eptcs}
\providecommand{\event}{SNR 2020} % Name of the event you are submitting to
\usepackage[T1]{fontenc}
\usepackage[latin9]{inputenc}
\usepackage{amsmath}
\usepackage{amssymb}
\usepackage{graphicx}

\usepackage{algorithm}
\usepackage{algorithmic}

\usepackage{amsthm}
\theoremstyle{definition}
\newtheorem{defn}{Definition}
\newtheorem{example}{Example}
\theoremstyle{plain}
\newtheorem{thm}{Theorem}
\newtheorem{prop}{Proposition}
\theoremstyle{remark}
\newtheorem{rem}{Remark}

\title{Interval centred form for proving stability of non-linear discrete-time systems}
\author{Auguste Bourgois
\institute{Lab-STICC\\ENSTA Bretagne\\Brest, France}
\institute{Forssea Robotics\\Paris, France}
\email{auguste.bourgois@ensta-bretagne.org}
\and
Luc Jaulin
\institute{Lab-STICC\\ENSTA Bretagne\\Brest, France}
\email{luc.jaulin@ensta-bretagne.fr}
}
\def\titlerunning{Interval centred form for proving stability of non-linear discrete-time systems}
\def\authorrunning{A. Bourgois \& L. Jaulin}
\begin{document}
\maketitle

\begin{abstract}
In this paper, we propose a new approach to prove stability
of non-linear discrete-time systems. After introducing the new concept of stability contractor, we show that the interval
centred form plays a fundamental role in this context and makes
it possible to easily prove asymptotic stability of a discrete
system. Then, we illustrate the principle of our approach through
theoretical examples. Finally, we provide two practical examples using
our method~: proving stability of a localisation system and
that of the trajectory of a robot.
\end{abstract}

\section{Introduction}
Proving properties of Cyber Physical Systems (CPS) is an important topic
that should be considered when designing reliable systems \cite{Asarin07,Frehse:08,Ramdani:Nedialkov11,taha:15:acumen}.
Among those properties, \emph{stability} is often wanted for a dynamical system~: achieving stability around a given setpoint in the state space of the latter is one of the aims of control theory. Proving stability of a dynamical system can be done rigorously \cite{jaulinsliding}, which is of major importance when applied to real-life systems. Indeed, a stable system is considered safe, since its behaviour is predictable.

Let us recall the definition of stability of a dynamical system.
Consider the non-linear discrete-time system
\begin{equation}
\begin{array}{ccl}
\mathbf{x}_{k+1} & = & \mathbf{f}(\mathbf{x}_{k})\end{array}\label{eq:seq1}
\end{equation}
According to Lyapunov's definition of stability \cite{fantoni:02,slotine91}, the system (\ref{eq:seq1}) is\emph{ stable} if
\begin{equation}
\forall\varepsilon>0,\:\exists\delta>0,\:\|\mathbf{x}_{0}\|<\delta\implies\forall k\geq0,\|\mathbf{x}_{k}\|<\varepsilon
\end{equation}
The system is\emph{ asymptotically stable} (or \emph{Lyapunov stable}) if there exists a neighbourhood of $\mathbf{0}$ such
that any initial state $\mathbf{x}_{0}$ taken in this neighbourhood
yields a trajectory which converges to $\mathbf{0}$,
\begin{equation}
\label{eq:asymptotic-stab}
\exists\delta>0,\:\|\mathbf{x}_{0}\|<\delta\implies\lim_{k\rightarrow\infty}\|\mathbf{x}_{k}\|=0
\end{equation}
Finally, the system is\emph{ exponentially stable} if for a given norm $\|\cdot\|$
\begin{equation}
\label{eq:exponential-stab}
\exists\delta>0,\:\exists\alpha\geq 0,\:\exists\beta\geq 0,\:\|\mathbf{x}_{0}\|<\delta\implies\forall k \geq 0,\:\|\mathbf{x}_{k}\|\leq\alpha\|\mathbf{x}_{0}\|e^{-\beta k}
\end{equation}

A classical method to prove stability of a system is to linearise the latter around $\mathbf{0}$
and check if the eigenvalues are inside the unit disk. However, due the inherent
uncertainties of a real-life system's model, no guarantee can
be obtained without using interval analysis \cite{Rump00}. The Jury
criterion \cite{JaulinBurger99} can also be used on the linearised
system in this context, but again, interval computation has to be
performed to get a proof of stability \cite{Rohn:stab:mat:96}. Moreover,
to our knowledge, none of the existing methods is able to give an
approximation for the neighbourhoods $\delta$ and $\varepsilon$ used
in Lyapunov's definition. Now, finding values for $\delta$ and
$\varepsilon$ is needed in practice, for instance to initialize algorithms
which approximate basins of attraction \cite{Lhommeau:Viability:Incinco07,ratschan_stab,SaintPierre02}
or reachable sets \cite{lemezomaze19}.

This paper proposes an original approach to prove Lyapunov stability of a non-linear discrete system,
but also to find values for $\delta$ and $\varepsilon$. It uses
the centred form, a classical concept in interval analysis \cite{Moore66}.
Moreover, it does not need the introduction of any Lyapunov function,
Jury criterion, linearisation, or any other classical tool used in
control theory.

Section \ref{sec:interval-analysis} briefly presents interval analysis and the notations used in the rest of this paper. Section \ref{sec:stab:ctr} introduces the new notion of \emph{stability contractor} and gives a theorem which explains why this concept is
useful for stability analysis. Section \ref{sec:centered-form} recalls
the definition of the centred form. It provides
a recursive version of the centred form in the case where the function
to enclose is the solution of a recurrence equation. It also shows
that the centred form can be used to build stability contractors.
Section \ref{sec:stability} shows that the approach is able to reach
a conclusion for a large class of stable systems and provides a procedure
for proving stability. Some examples are given to illustrate graphically
the properties of the approach. Section \ref{sec:conclusion} concludes
the paper and proposes some perspectives.

\section{Introduction to interval analysis}
\label{sec:interval-analysis}
The method presented in this paper is based on interval analysis. In short, interval analysis is a field of mathematics where intervals, \textit{i.e.} connected subsets of $\mathbb{R}$, are used instead of real numbers to perform computations. Doing so allows to enclose all types of uncertainties (from floating-point round-off to modelling errors) of a system and therefore yield a guaranteed enclosure for the solution of a problem related to that system. In this section, we briefly introduce the notations and important concepts used later in this paper. More details about interval analysis and its applications can be found in \cite{JaulinBook01,moore2009introduction}.

An interval is a set delimited by a lower bound $x^-$ and an upper bound $x^+$ such that $x^-\leq x^+$~:\[[x]=[x^-,x^+]\]
Intervals can be stacked into vectors, and are thus denoted by \[[\mathbf{x}]=\left([x_1],[x_2],\dots,[x_n]\right)\]
We write $[u,v]^{\times n}$ the interval vector of size $n$, all the components of which are equal to $[u,v]$.

Vectors of intervals are often called \emph{boxes} or \emph{interval vectors}. The set of axis-aligned boxes of $\mathbb{R}^n$ is denoted by $\mathbb{IR}^n$. Similarly, interval vectors can be concatenated into \emph{interval matrices}.

Intervals, interval vectors (or matrices) can be multiplied by a real $\lambda$ as such~:
\[\lambda[x^-,x^+]=\begin{cases}
[\lambda x^-, \lambda x^+] & \text{if}\;\lambda\geq0\\
[\lambda x^+, \lambda x^-] & \text{if}\;\lambda<0\\
\end{cases}\]

We denote by $w\left([x]\right)$ the width of $[x]$~:\[w\left([x]\right)=x^+ - x^-\]
The width of an interval vector $[\mathbf{x}]$ is given by
\[w\left([\mathbf{x}]\right)=\underset{i}{\max}\left(w\left([x_i]\right)\right)\]

The absolute value of an interval $[x]$ is \[\left\lvert[x]\right\rvert=\max\left(\left\lvert x^-\right\rvert, \left|x^+\right|\right)\]
And the norm of an interval vector $[\mathbf{x}]$ is defined as follows \[\left\lVert[\mathbf{x}]\right\rVert=\underset{i}{\max}\left\lvert[x_i]\right\rvert\]

Later in this paper, we will use the following implication
\begin{equation}
\label{eq:implication}
\forall [\mathbf{a}],\,[\mathbf{b}] \in \mathbb{IR}^n,\,[\mathbf{a}]\subset[\mathbf{b}]\implies\left\lVert[\mathbf{a}]\right\rVert<\left\lVert[\mathbf{b}]\right\rVert
\end{equation}

The usual arithmetic operators ($+,-,\times,/$) can be defined over intervals (as well as boxes and interval matrices). Operations involving intervals, interval vectors and interval matrices can therefore be naturally deduced from their real counterpart.

Extending a real function to intervals (and equivalently to interval vectors/matrices) can also be achieved as follows~:
\[f\left([x]\right)=\left\{f(x),x\in[x]\right\}\]

Except in trivial cases, $f\left([x]\right)$ usually cannot be written as an interval, whence the use of \emph{inclusion functions}. An inclusion function $[f]\left([x]\right)$ associated with $f\left([x]\right)$ yields an interval (or interval vector/matrix) enclosing the set $f\left([x]\right)$~:\[f\left([x]\right)\subset [f]\left([x]\right)\]
$[f]$ is said to be \emph{minimal} if $[f]\left([x]\right)$ is the smallest interval (or interval vector/matrix) enclosing the set $f\left([x]\right)$ (see Figure \ref{fig:inclusion}). The minimal inclusion function associated with $f$ will be denoted by $[f\left([x]\right)]$.
\begin{figure}[h]
	\centering
	\includegraphics[width=0.6\textwidth]{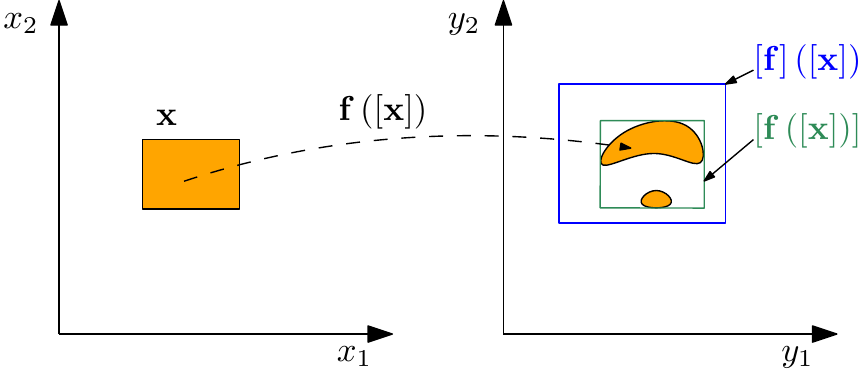}
	\caption{Interval function, inclusion function \& minimal inclusion function}
	\label{fig:inclusion}
\end{figure}
An inclusion function is said to be \emph{natural} when it is expressed by replacing its variables and its elementary functions and operators by their interval counterparts.
\begin{example}
	Consider the function $f:\mathbb{R}^2\rightarrow\mathbb{R}$ such that for $\mathbf{x}=\left(x_1,x_2\right)\in\mathbb{R}^2$ \[f(\mathbf{x})=\sin\left(x_1\right)+\exp\left(x_2\right)\]
	The natural inclusion function $[f]$ of $f$ is~:\[[f]([\mathbf{x}])=[\sin]\left([x_1]\right)+[\exp]\left([x_2]\right)\]
	where $[\sin]$ and $[\exp]$ are the inclusion functions of $\sin$ and $\exp$.
\end{example}

\section{Stability contractor\label{sec:stab:ctr}}
In this section, we present the concept of \emph{stability contractor}, a tool that can be used to rigorously prove the stability of a dynamical system. The rigour of the method comes from the use of interval analysis.

The following new definition adapts the definition of a contractor,
as given in \cite{Chabert09}, to stability analysis.
\begin{defn}
	\label{def:stab-contract}
	Consider a box $[\mathbf{x}_{0}]$ of $\mathbb{R}^{n}$. A\emph{ stability
		contractor} $\boldsymbol{\Psi}:\mathbb{IR}^{n}\rightarrow\mathbb{IR}^{n}$
	of rate $\alpha<1$ is an operator which satisfies
	\begin{equation}
	\begin{array}{ccc}
	(i) & [\mathbf{a}]\subset[\mathbf{b}]\implies\boldsymbol{\Psi}([\mathbf{a}])\subset\boldsymbol{\Psi}([\mathbf{b}]) & 
	\text{(monotonicity)}\\
	(ii) & \boldsymbol{\Psi}([\mathbf{a}])\subset[\mathbf{a}] & \text{(contractance)}\\
	(iii) & \boldsymbol{\Psi}(\mathbf{0})=\mathbf{0} & \text{(equilibrium)}\\
	(iv) & \boldsymbol{\Psi}([\mathbf{a}])\subset\alpha\cdot[\mathbf{a}]\implies\forall k \geq 1,\,\boldsymbol{\Psi}^{k}([\mathbf{a}])\subset\alpha^{k}\cdot[\mathbf{a}] & \text{(convergence)}
	\end{array}\label{eq:def1}
	\end{equation}
	for all boxes $[\mathbf{a}],[\mathbf{b}]$ inside $[\mathbf{x}_{0}]$. For $k\geq1$ ${\Psi}^{k}$ denotes the iterated function $\underset{k}{\underbrace{\boldsymbol{\Psi}\circ\dots\circ\boldsymbol{\Psi}}}$. If $k=0$, $\boldsymbol{\boldsymbol{\Psi}}^{0}$ denotes the identity function.
\end{defn}

\begin{example}
	If $[x]$ is an interval, the operator $[x]\mapsto[x]\cap0.9\cdot[x]$
	is a stability contractor whereas the operator $[x]\mapsto0.9\cdot[x]$
	is not.
\end{example}

\begin{prop}
	\label{prop:stab-contract}
	If $\boldsymbol{\Psi}$ is a stability contractor of rate $\alpha<1$, then
	\begin{equation}
	\boldsymbol{\Psi}([\mathbf{x}])\subset\alpha\cdot[\mathbf{x}]\implies\lim_{k\rightarrow\infty}\boldsymbol{\Psi}^{k}([\mathbf{x}])=\mathbf{0}
	\end{equation}
\end{prop}
\begin{proof}
	Let $[\mathbf{x}]\in\mathbb{IR}^n$ and $\boldsymbol{\Psi}$ denote a stability contractor such that $\boldsymbol{\Psi}([\mathbf{x}])\subset\alpha\cdot[\mathbf{x}]$.
	Then, according to Equation (\ref{eq:def1}, iv),
	\begin{align*}
	\boldsymbol{\Psi}([\mathbf{x}])\subset\alpha\cdot[\mathbf{x}]&\implies\boldsymbol{\Psi}^{k}([\mathbf{x}])\subset\alpha^{k}\cdot[\mathbf{x}]\\
	&\implies\lim_{k\rightarrow\infty}\boldsymbol{\Psi}^{k}([\mathbf{x}])=\mathbf{0}
	\end{align*}
\end{proof}

A consequence of this proposition is that getting a stability contractor
allows us to prove stability of a system without performing an
infinitely long set-membership simulation. It suffices to have one
box $[\mathbf{x}]\ni\mathbf{0}$ such that $\boldsymbol{\Psi}([\mathbf{x}]))\subset\alpha\cdot[\mathbf{x}]$,
$\alpha<1$ to conclude that the system
is asymptotically stable for all $\mathbf{x}\in[\mathbf{x}]$. In other words, the system is proven to be Lyapunov stable in the neighbourhood $[\mathbf{x}_0]$.

Now, if one can build a stability contractor $\boldsymbol{\Psi}$ for a given dynamical system, then proving Lyapunov stability of the latter for a given initial condition $[\mathbf{x}]$ comes down to applying proposition \ref{prop:stab-contract}. Building such a contractor is addressed in the next section.

\section{Centred form of an interval function}
\label{sec:centered-form}
The aim of this section is to present a general method for building a stability contractor for discrete dynamical systems. It is based on the concept on centred form, proposed by Moore in \cite{Moore79}. Additionally, it proposes an algorithm to compute iteratively the centred form of an iterated function.

Consider a function $\mathbf{f}:\mathbb{R}^{n}\rightarrow\mathbb{R}^{n}$, with
a Jacobian matrix $\mathbf{J}(\mathbf{x})=\frac{d\mathbf{f}}{d\mathbf{x}}(\mathbf{x})$.
Consider a box $[\mathbf{x}]$ and one point $\bar{\mathbf{x}}$ in
$[\mathbf{x}]$. For simplicity, and without loss of generality, we assume that $\bar{\mathbf{x}}=\mathbf{0}\in[\mathbf{x}]$ and that $\mathbf{f}(\mathbf{0})=\mathbf{0}$.
\begin{rem}
	\label{rem:centering-probem}
	If this condition is not satisfied,\emph{ i.e.} $\bar{\mathbf{x}}\neq\mathbf{0}$
	or $\mathbf{f}(\mathbf{0})\neq\mathbf{0}$, the problem $\mathbf{y}=\mathbf{f}(\mathbf{x}),\mathbf{x}\in[\mathbf{x}]$
	can be transformed into an equivalent problem satisfying the latter~:
	\begin{equation}
	\begin{array}{ccc}
	\mathbf{y}=\mathbf{f}(\mathbf{x}) & \Leftrightarrow & \underset{=\mathbf{z}}{\underbrace{\mathbf{y}-\mathbf{f}(\bar{\mathbf{x}})}}=-\mathbf{f}(\bar{\mathbf{x}})+\mathbf{f}(\underset{=\mathbf{p}}{\underbrace{\mathbf{x}-\bar{\mathbf{x}}}}+\bar{\mathbf{x}})\\
	& \Leftrightarrow & \mathbf{z}=\underset{=\mathbf{g}(\mathbf{p})}{\underbrace{-\mathbf{f}(\bar{\mathbf{x}})+\mathbf{f}(\mathbf{p}+\bar{\mathbf{x}})}}
	\end{array}
	\end{equation}
	\textit{i.e.}
	\begin{equation}
	\left\{ \begin{array}{ccc}
	\mathbf{y} & = & \mathbf{z}+\mathbf{f}(\bar{\mathbf{x}})\\
	\mathbf{z} & = & \mathbf{g}(\mathbf{p})\\
	\mathbf{p} & = & \mathbf{x}-\bar{\mathbf{x}}
	\end{array}\right.
	\end{equation}
\end{rem}
Now, consider the new problem $\mathbf{z}=\mathbf{g}(\mathbf{p}),\mathbf{p}\in[\mathbf{p}]$ where $[\mathbf{p}]=[\mathbf{x}]-\bar{\mathbf{x}}$. Since $\mathbf{g}(\mathbf{0})=\mathbf{0}$ and $\mathbf{0}\in[\mathbf{p}]$, the new problem satisfies the condition stated above.

Let us recall the definition of the centred form,as given by \cite{Moore79}~:
\begin{defn}
	The centred form $\mathbf{f}_{c}$ associated to
	the function $\mathbf{f}$ is given by
	\begin{equation}
	\label{eq:centred-form}
	[\mathbf{f}_{c}]([\mathbf{x}])=\left(\left[\mathbf{J}\right]([\mathbf{x}])\right)\cdot[\mathbf{x}]
	\end{equation}
	where $\left[\mathbf{J}\right]$ is the natural extension of $\mathbf{J}$.
\end{defn}
\begin{prop}
	\label{prop:order:fc}
	According to \cite{Moore79}, for all $[\mathbf{x}]\in\mathbb{IR}^n$,
	\begin{equation}
	\mathbf{f}([\mathbf{x}])\subset[\mathbf{f}_{c}]([\mathbf{x}])
	\end{equation}
	
	Additionally, the centred form tends towards the minimal inclusion function when $[\mathbf{x}]$ is sufficiently small~:
	\begin{equation}
	w\left([\mathbf{f}_{c}]\left([\mathbf{x}]\right)\right)-w\left([\mathbf{f}\left([\mathbf{x}]\right)]\right)=o\left(w\left([\mathbf{x}]\right)\right)\;\text{as}\;w\left([\mathbf{x}]\right)\rightarrow 0 
	\end{equation}
\end{prop}

The statements of proposition \ref{prop:order:fc} are illustrated on figure \ref{fig:Centered-form}.

\begin{figure}[h]
	\centering
	\includegraphics[width=0.7\columnwidth]{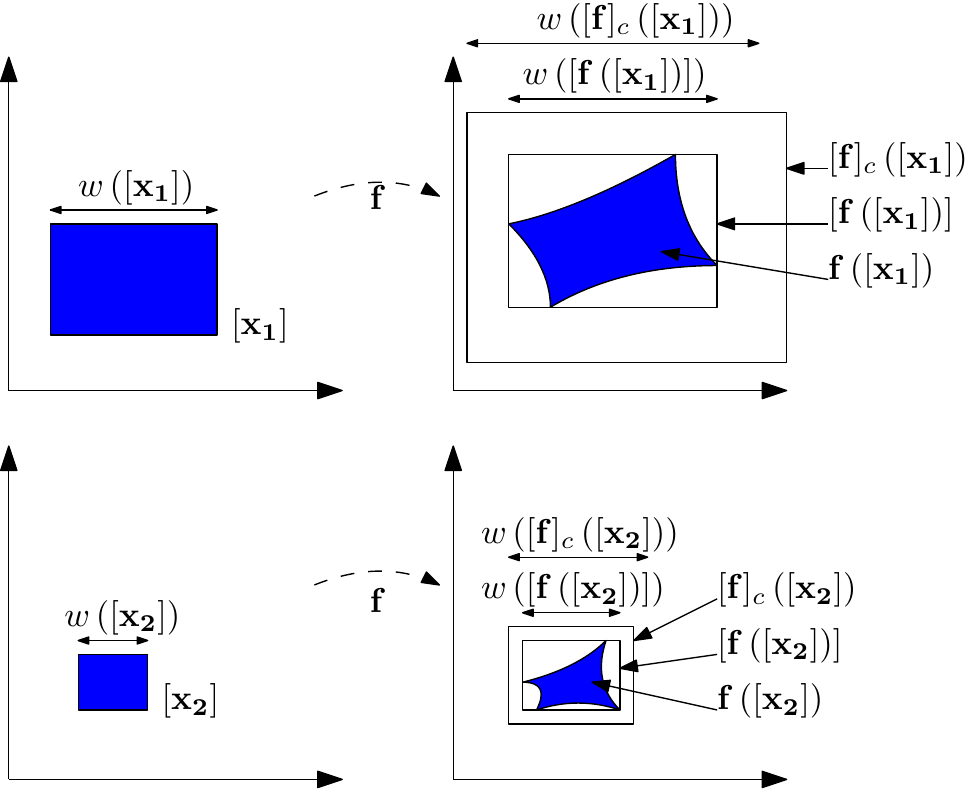}
	\caption{Illustration of the centred form extension of an interval function\label{fig:Centered-form}}
\end{figure}

\begin{thm}
	\label{thm:seq:centered0}Consider a function $\mathbf{f}$ with $\mathbf{f}(\mathbf{0})=\mathbf{0}$
	and with Jacobian matrix $\mathbf{J}(\mathbf{x})=\frac{d\mathbf{f}}{d\mathbf{x}}(\mathbf{x})$.
	Denote by $[\mathbf{f}]$ and $[\mathbf{J}]$ the natural inclusion
	functions for $\mathbf{f}$ and\textbf{ $\mathbf{J}$}. The centred
	form $[\mathbf{f}_{c}^{k}]$ associated to $\mathbf{f}^{k}=\mathbf{f}\circ\mathbf{f}\circ\dots\circ\mathbf{f}$
	is given by the following sequence
	\begin{equation}
	\begin{array}{ccc}
	[\mathbf{z}](0) & = & [\mathbf{x}]\\{}
	[\mathbf{A}](0) & = & \text{Id}\\{}
	[\mathbf{z}](k) & = & [\mathbf{f}]\left([\mathbf{z}](k-1)\right)\\{}
	[\mathbf{A}](k) & = & [\mathbf{J}]([\mathbf{z}](k-1))\cdot[\mathbf{A}](k-1)\\{}
	[\mathbf{f}_{c}^{k}]([\mathbf{x}]) & = & [\mathbf{A}](k)\cdot[\mathbf{x}]
	\end{array}\label{eq:th1}
	\end{equation}
\end{thm}
\begin{proof}
	Consider a function $\mathbf{f}:\mathbb{R}^n\rightarrow\mathbb{R}^n$ such that $\mathbf{f}(\mathbf{0})=\mathbf{0}$, and denote $\mathbf{J}=\frac{d\mathbf{f}}{d\mathbf{x}}$ its Jacobian matrix.
	Write $\mathbf{f}^k$ the $k$-th iterated function of $\mathbf{f}$. Then for all $\mathbf{x}\in\mathbb{R}^n$, the Jacobian matrix of $\mathbf{f}^k$ is given by~:
	\begin{equation}
	\label{eq:decomposition}
	\dfrac{d\mathbf{f}^{k}}{d\mathbf{x}}(\mathbf{x})=\dfrac{d(\mathbf{f}\circ\mathbf{f}^{k-1})}{d\mathbf{x}}(\mathbf{x})=\dfrac{d\mathbf{f}}{d\mathbf{x}}(\mathbf{f}^{k-1}(\mathbf{x}))\cdot\dfrac{d\mathbf{f}^{k-1}}{d\mathbf{x}}(\mathbf{x}).
	\end{equation}
	Define $\mathbf{z}(k)=\mathbf{f}^{k}(\mathbf{x})$ and $\mathbf{A}(k)=\frac{d\mathbf{f}^{k}}{d\mathbf{x}}(\mathbf{x}).$
	Since for all $\mathbf{x}$,  $\mathbf{J}(\mathbf{x})=\frac{d\mathbf{f}}{d\mathbf{x}}(\mathbf{x})$,
	we get
	\begin{equation}
	\mathbf{A}(k)=\mathbf{J}(\mathbf{z}(k-1))\cdot\mathbf{A}(k-1).
	\end{equation}
	Now, since $\mathbf{f}(\mathbf{0})=\mathbf{0}$, it follows that $\mathbf{f}^k\left(\mathbf{0}\right)=\mathbf{0}$. Finally, substituting Equation (\ref{eq:decomposition}) into Equation (\ref{eq:centred-form}) yields 
	\begin{align*}
	[\mathbf{f}_{c}^{k}]([\mathbf{x}])&=\left[\dfrac{d\mathbf{f}^k}{d\mathbf{x}}\right]\left([\mathbf{x}]\right)\cdot[\mathbf{x}]\\
	&=[\mathbf{A}](k)\cdot[\mathbf{x}]
	\end{align*}
\end{proof}
\begin{rem}
	\label{rem:pessimism}
	From Proposition \ref{prop:order:fc}, for a given $k$, we have
	\begin{equation}
	\lim_{w([\mathbf{x}])\rightarrow0}\frac{w([\mathbf{f}_{c}^{k}]([\mathbf{x}]))-w([\mathbf{f}^{k}([\mathbf{x}])])}{w([\mathbf{x}])}=0
	\end{equation}
	which means that $[\mathbf{f}_{c}^{k}]$ tends towards the minimal inclusion function when $w([\mathbf{x}])\rightarrow0$.
	Now, for a given $[\mathbf{x}]$, even very small, we generally observe the following~:
	\begin{equation}
	\lim_{k\rightarrow\infty}\frac{w([\mathbf{f}_{c}^{k}]([\mathbf{x}]))-w([\mathbf{f}^{k}([\mathbf{x}])])}{w([\mathbf{x}])}=\infty.
	\end{equation}
	In other words, the pessimism introduced by the centred form increases
	with $k$.
\end{rem}
\begin{thm}
	\label{th:fc:stab}Consider a function $\mathbf{f}$ with $\mathbf{f}(\mathbf{0})=\mathbf{0}$.
	If there exist a box $[\mathbf{x}_{0}]\ni\mathbf{0}$ and a real $\alpha<1$ such that $[\mathbf{f}_{c}]([\mathbf{x}_{0}])\subset\alpha\cdot[\mathbf{x}_{0}]$
	then the interval operator $\boldsymbol{\Psi}_{\mathbf{f}}([\mathbf{x}])=[\mathbf{f}_{c}]([\mathbf{x}])$
	is a stability contractor inside $[\mathbf{x}_{0}]$.
\end{thm}
\begin{proof}
	Properties (\ref{eq:def1}, i), (\ref{eq:def1}, ii) and (\ref{eq:def1}, iii) of definition \ref{def:stab-contract} are easily
	checked. Now, let us prove property (\ref{eq:def1}, iv) by induction.
	
	Let $\alpha<1$ and $[\mathbf{x}]\ni\mathbf{0}$.	Property (\ref{eq:def1}, iv) states that \[[\mathbf{f}_c]\left([\mathbf{x}]\right)\subset\alpha[\mathbf{x}]\implies\forall k \geq 0,\, [\mathbf{f}_c]^k\left([\mathbf{x}]\right)\subset\alpha^k\cdot[\mathbf{x}]\]
	where $[\mathbf{f}_c]^k=\underset{k}{\underbrace{[\mathbf{f}_{c}]\circ\dots\circ[\mathbf{f}_{c}]}}$.
	Since $[\mathbf{x}]\subset\alpha^0[\mathbf{x}]$, the property holds for $k=0$.
	
	Let us define the sequence \[[\mathbf{w}_{k+1}]=[\mathbf{f}_c]\left([\mathbf{w}_k]\right)\]
	Now, assume that the property also holds for a $k>0$, \textit{i.e.}\[[\mathbf{w}_k]\subset\alpha^k\cdot[\mathbf{x}]\]
	Let us show that \[[\mathbf{w}_{k+1}]\subset\alpha^{k+1}\cdot[\mathbf{x}]\]
	\[
	\begin{array}{ccll}
	[\mathbf{w}_{k+1}] & = & [\mathbf{f}_{c}]([\mathbf{w}_k])\\
	& = & [\mathbf{J}]\left([\mathbf{w}_k]\right)\cdot[\mathbf{w}_k]\\
	& \subset & [\mathbf{J}]\left([\mathbf{x}]\right)\cdot[\mathbf{w}_k] & \text{since \ensuremath{\left[\mathbf{J}\right]} is inclusion monotonic and \ensuremath{[\mathbf{w}_k]\subset\alpha^k\cdot[\mathbf{x}]\subset[\mathbf{x}]}}\\
	& \subset & [\mathbf{J}]\left([\mathbf{x}]\right)\cdot(\alpha^k\cdot[\mathbf{x}]) & \text{since \ensuremath{[\mathbf{w}_k]\subset\alpha^k[\mathbf{x}]}}\\
	& = & \alpha^k\cdot\left[\mathbf{J}\right]([\mathbf{x}])\cdot[\mathbf{x}]\\
	& = & \alpha^k\cdot[\mathbf{f}_c]\left([\mathbf{x}]\right)\\
	& \subset & \alpha^k\cdot(\alpha\cdot[\mathbf{x}])\\
	& \subset & \alpha^{k+1}\cdot[\mathbf{x}] & \text{since \ensuremath{[\mathbf{f}_c]\left([\mathbf{x}]\right)\subset\alpha[\mathbf{x}]}}
	\end{array}
	\]
\end{proof}

Theorem \ref{th:fc:stab} shows that the centred form can be used in a
very general way to build stability contractors. Other integration
methods such as the one proposed by Lohner \cite{Lohner87} or the ones based on affine
forms \cite{Andrade94} do not have this property, even if they usually yield a tighter enclosure of the image set of a function.

\section{Proving stability using the centred form\label{sec:stability}}

In this section, we show how the centred form can be used to prove stability of non-linear dynamical systems.

\subsection{Method}

Consider the system described by Equation (\ref{eq:seq1}), where $\mathbf{f}\left(\mathbf{0}\right)=\mathbf{0}$.

In short, the methods consists in finding an initial box $[\mathbf{x}]$ in the state space of the system such that the centred form of $\mathbf{f}$ is a stability contractor. This implies that iterating this centred form onto that initial box will converge towards $\mathbf{0}$. Since the centred form is an inclusion function for $\mathbf{f}$, that implies that the system will also converge towards $\mathbf{0}$.

More specifically, if for a given box $[\mathbf{x}]\ni\mathbf{0}$ with width $\delta$ there exist $q\geq1$ and $\alpha<1$ such that $[\mathbf{f}_c^q]\left([\mathbf{x}]\right)\subset\alpha\cdot[\mathbf{x}]$, then $[\mathbf{f}_c^q]$ is a stability contractor of rate $\alpha$ (according to theorem \ref{th:fc:stab}). Now, since proposition \ref{prop:stab-contract} states \[\lim_{k\rightarrow\infty}[\mathbf{f}_c^q]^{k}([\mathbf{x}])=\mathbf{0}\]
and since proposition \ref{prop:order:fc} asserts that \[\forall\mathbf{x}\in[\mathbf{x}],\:\mathbf{f}^q\left(\mathbf{x}\right)\in\mathbf{f}^q\left([\mathbf{x}]\right)\subset[\mathbf{f}^q_c]\left([\mathbf{x}]\right)\]
it follows that\[\forall\mathbf{x}\in[\mathbf{x}],\:\lim_{k\rightarrow\infty}\left(\mathbf{f}^q\right)^{k}(\mathbf{x})=\mathbf{0}\]
which implies asymptotic stability of the system (see Equation \ref{eq:asymptotic-stab}).

Furthermore, let us define the sequence
\begin{equation}
\label{eq:sequence}
[\mathbf{x}_{k+1}]=[\mathbf{f}^q_c]^k\left([\mathbf{x}]\right)
\end{equation}
Thus, substituting Equation (\ref{eq:sequence}) in property (\ref{eq:def1}, iv) and applying Equation (\ref{eq:implication}) yields
\begin{align*}
\left\lVert[\mathbf{x}_{k+1}]\right\rVert&\leq\alpha^k\left\lVert[\mathbf{x}]\right\rVert\\
&\leq\left\lVert[\mathbf{x}]\right\rVert e^{\ln\left(\alpha\right)k}\\
&\leq\left\lVert[\mathbf{x}]\right\rVert e^{-\beta k}
\end{align*}
where $\beta = -\ln\left(\alpha\right) > 0$. This implies exponential stability of the system (see Equation \ref{eq:exponential-stab}).

Our method is summarized by algorithm \ref{alg:method}. It takes as inputs the function $\mathbf{f}$ describing the system, the initial box $[\mathbf{x}]$, and a maximum number of iterations $N$. The latter is necessary, as stated in remark \ref{rem:pessimism} to avoid looping indefinitely. Usually, $N$ does not need to be larger than $10$, since the iterated centred form tends to diverge with the number of iterations.
\begin{algorithm}[h]
	\begin{algorithmic}
		\REQUIRE $\mathbf{f}$, $[\mathbf{x}]$, $N$
		\ENSURE \textbf{true} if the system is stable, \textbf{false} if undetermined
		\STATE $[\mathbf{A}]\leftarrow\mathbf{I}$
		\STATE $[\mathbf{z}]\leftarrow[\mathbf{x}]$
		\FOR{$i=1$ \TO $N$}
			\STATE $[\mathbf{A}]\leftarrow[\mathbf{J}]\left([\mathbf{z}]\right)\cdot[\mathbf{A}]$
			\STATE $[\mathbf{z}]\leftarrow[\mathbf{f}]\left([\mathbf{z}]\right)$
			\STATE $[\mathbf{x}_i]\leftarrow[\mathbf{A}]\cdot[\mathbf{x}]$
			\IF {$[\mathbf{x}_i]\subset[\mathbf{x}]$}
				\RETURN \textbf{true}
			\ENDIF
		\ENDFOR
		\RETURN \textbf{false}
	\end{algorithmic}
	\caption{Centred form based stability contractor}
	\label{alg:method}
\end{algorithm}

It is important to remember that if the previous algorithm yields a \textbf{false} value, that does not necessarily mean that the system is unstable. Indeed, bisecting the initial box $[\mathbf{x}]$ into smaller boxes and running the algorithm on them could prove stability of the system in the neighbourhood defined by those boxes.

\subsection{Completeness of the method}
\label{sec:completeness}
Now, let us show that whenever a system actually is exponentially stable, our method will be able to prove it. Given $\eta>0$, we denote by $\mathcal{B}_{\eta}$ the set of all hypercubes $[\mathbf{x}]$ centred at $\mathbf{0}$ and such that $w([\mathbf{x}])\leq\eta$.
\begin{prop}
	\label{prop:exp_stab_psi}
	If the system is exponentially stable around $\mathbf{0}$, then
	\[
	\exists\eta>0,\:\forall[\mathbf{x}]\in\mathcal{B}_{\eta},\:\exists k>0,\:\exists\alpha<1,\:[\mathbf{f}_{c}^{k}]([\mathbf{x}])\subset\alpha\cdot[\mathbf{x}]
	\]
\end{prop}
\begin{proof}
	Let us assume exponential stability of the system described by $\mathbf{f}$ around $\mathbf{0}$, \textit{i.e.} there exists a neighbourhood $\mathcal{N}$ of $\mathbf{0}$, such that
	\[
	\exists\beta,\:0<\beta<1,\:\exists k>0,\:\forall\mathbf{x}\in\mathcal{N},\:\|\mathbf{f}^{k}(\mathbf{x})\|_{\infty}<\beta\|\mathbf{x}\|_{\infty}
	\]
	This property translates into
	\[
	[\mathbf{f}^{k}([\mathbf{x}])]\subset\beta\cdot[\mathbf{x}]
	\]
	for all cubes $[\mathbf{x}]$ in $\mathcal{N}$ centred in $\mathbf{0}$,
	where $[\mathbf{f}^{k}([\mathbf{x}])]$ is the smallest box which
	contains the set $\mathbf{f}^{k}([\mathbf{x}])$. Take one of these
	cubes $[\mathbf{x}]$ and denote by $\eta$ its width $w\left([\mathbf{x}]\right)$. If $\eta$ is sufficiently small, the pessimism of the centred form becomes arbitrarily
	small and we have
	\[
	\exists\alpha,\:\beta\leq\alpha<1,\:[\mathbf{f}_{c}^{k}]([\mathbf{x}])\subset\alpha[\mathbf{x}]
	\]
\end{proof}

\section{Application}
\subsection{Example 1~: Proving stability of the logistic map}
The aim of this example is to illustrate the section \ref{sec:completeness}.
The logistic map is a simple recurrence relation which behaviour can
be highly complex or chaotic~:
\begin{equation}
x_{k+1}=\rho\cdot x_{k}\cdot\left(1-x_{k}\right)\label{eq:logistic}
\end{equation}
The parameter $\rho$ influences the dynamics of the system. For $\rho=2.4$,
the system is stable, but shows an oscillating behaviour around its
equilibrium point (see figure \ref{fig:ex1:logistic-no-interval}).
\begin{figure}[h]
	\centering
	\includegraphics[width=0.7\columnwidth]{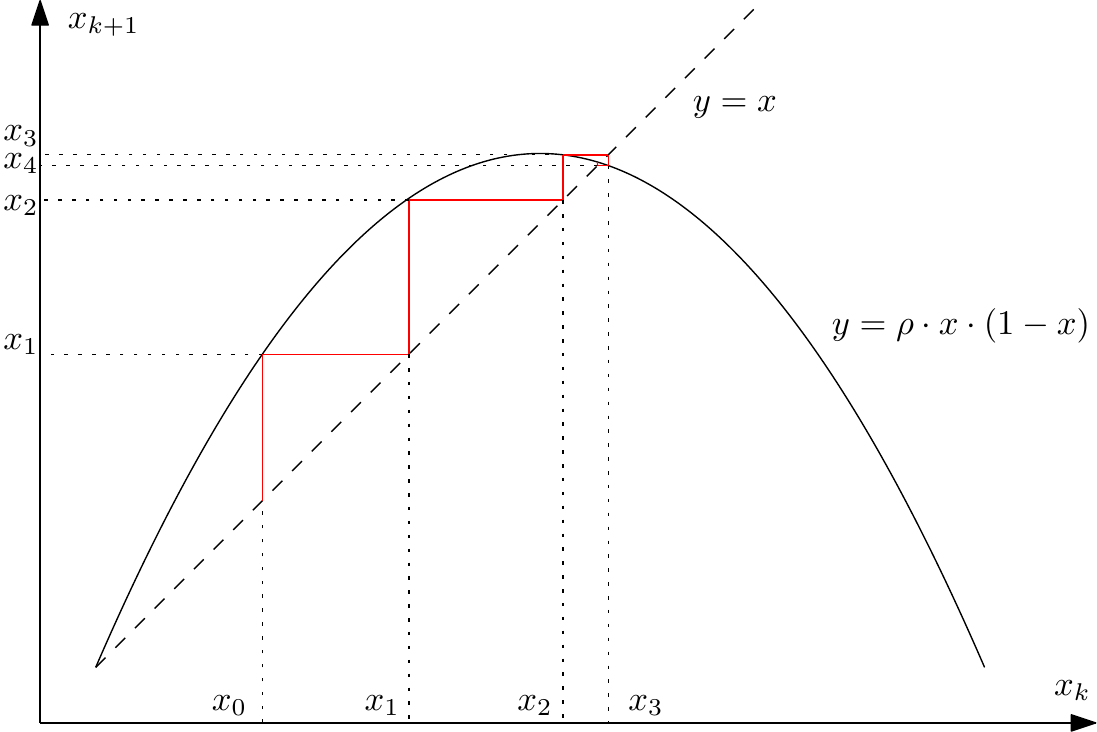}
	\caption{Behaviour of the logistic map for $\rho=2.4$}
	\label{fig:ex1:logistic-no-interval}
\end{figure}

Note that the equilibrium point is not $\mathbf{0}$
and we thus need to centre our problem to apply the stability method
as described in this paper (see remark \ref{rem:centering-probem}).

Figure \ref{fig:ex1:logistic} displays the behaviour of our algorithm for an initial box
$[x_{0}]=\left[0.577,\:0.585\right]$. Even if $[x_{1}]$ is not contained
in $[x_{0}]$, we have $[x_{2}]\subset[x_{0}]$. We have thus shown
that for an initial state chosen in $[x_{0}]$, the trajectory will
converge towards the stable point.
\begin{figure}[h]
	\centering
	\includegraphics[width=0.9\columnwidth]{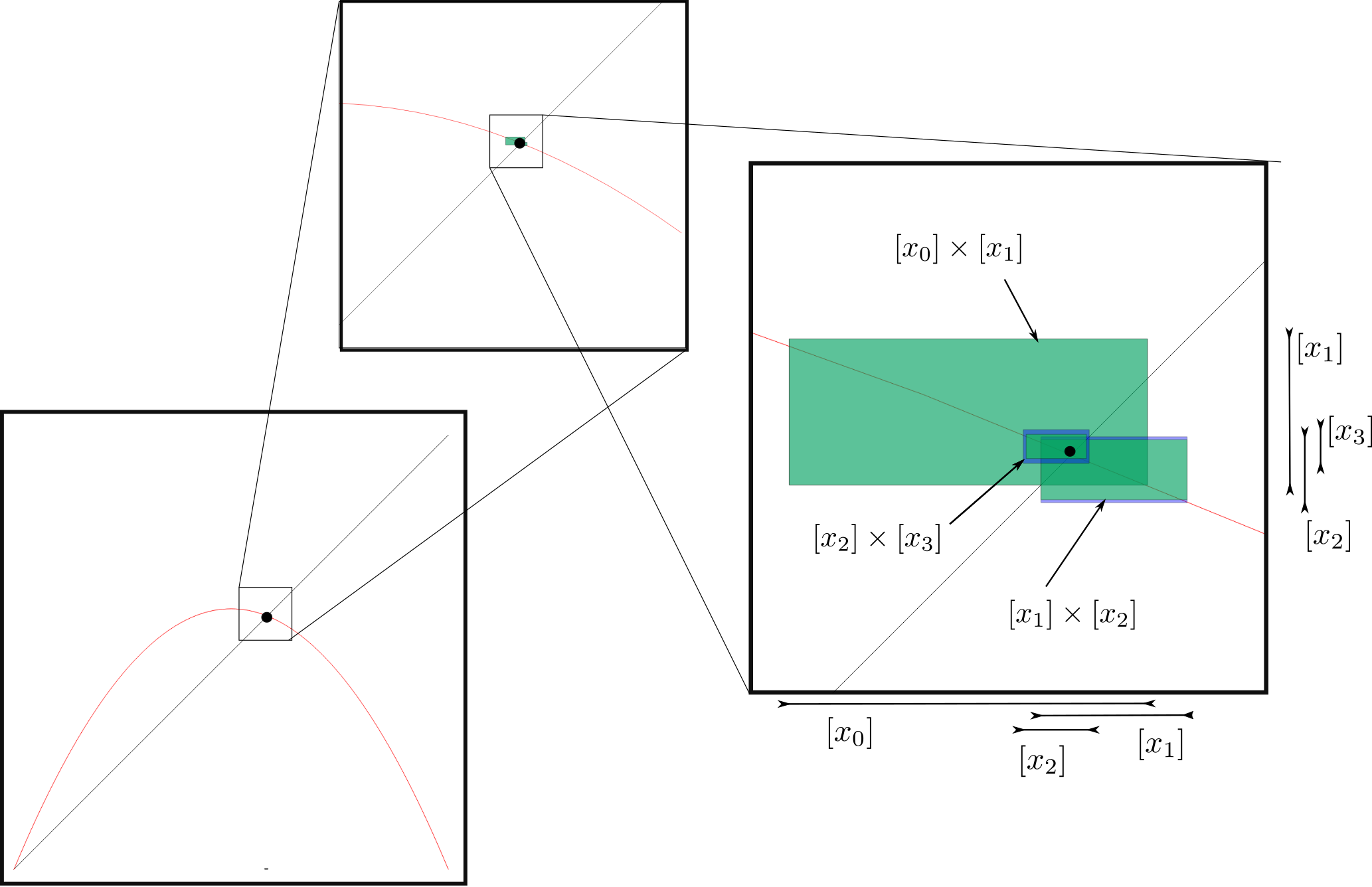}
	\caption{Stability contractor proving stability of the logistic map. The initial box is small, whence the 2 enlargements symbolised by the connected black squares.}
	\label{fig:ex1:logistic}
\end{figure}

\subsection{Example 2~: Proving stability of a three dimensional map}
The aim of this example is to illustrate the steps of algorithm \ref{alg:method} in a higher dimensional problem.

Let us consider the following discrete system:
\begin{equation}
\begin{array}{ccc}
\mathbf{x}_{k+1} & = & 0.8\cdot\mathbf{R}\left(\frac{\pi}{6}+x_{1},\:\frac{\pi}{4}+x_{2},\:\frac{\pi}{3}+x_{3}\right)\cdot\mathbf{x}_{k}\\
& = & \mathbf{f}\left(\mathbf{x}_{k}\right)
\end{array}\label{eq:3dsystem}
\end{equation}
where $\mathbf{R}\left(\varphi,\theta,\psi\right)$ is the rotation
matrix parametrized by roll ($\varphi$, around axis $X$), pitch
($\theta$, around $Y$) and yaw ($\psi$, around $Z$) angles.

With our approach, we show that the discrete system is stable, as
depicted by Figure \ref{fig:ex2:cube:euler} which is a projection
of the 3-dimensional system across steps of our algorithm. With $[\mathbf{x}_{0}]=[-\varepsilon,\varepsilon]^{\times3}$,
where $\varepsilon=0.004$, the algorithm needs 3 iterations to get
captured inside the initial box. We thus have proved the stability
of the system. The blue sets correspond to the image
of $[\mathbf{x}_{0}]$ by $\mathbf{f},\:\mathbf{f}^{2},\:\mathbf{f}^{3}$
and has been obtained using a Monte-Carlo method for visualization
purposes. The point in the centre corresponds to the origin $\mathbf{0}$.

\begin{figure}[h]
	\centering
	\includegraphics[width=1\columnwidth]{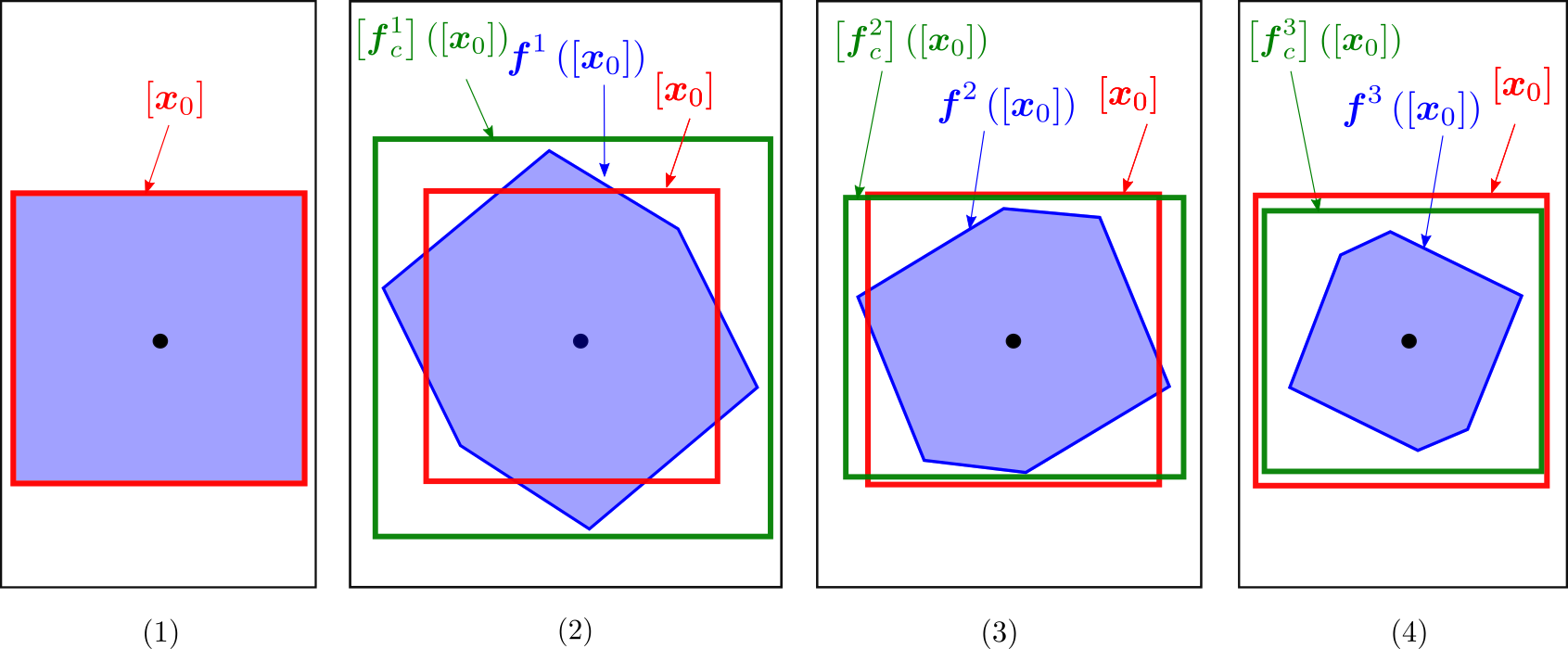}
	\caption{View in the $(x_{1},x_{2})$-space of the effect of the stability contractor
		acting on the $(x_{1},x_{2},x_{3})$-space\label{fig:ex2:cube:euler}}
\end{figure}

Considering the shape of the set $\mathbf{f}^{k}([\mathbf{x}_{0}])$,
we understand that zonotope-based methods \cite{Combastel05,JianWan07}
or Lohner integration methods \cite{Lohner87,wilczak2011} could get a stronger convergence.
However, these efficient operators cannot be used to prove the stability
except if we are able to prove that they correspond to a stability
contractor, which is not the case yet.

\subsection{Example 3~: Validation of a localisation system}
The aim of this example is to illustrate how our method could be applied to prove stability of an existing localisation system, before integrating the latter in a robot for example.
\begin{rem}
	Proving the stability of such a localisation system is of importance in robotics. Indeed, the commands are usually computed using the robot's state, estimated by the localisation system. If the latter is not stable, \textit{i.e.} it does not converge towards the actual position of the robot, then control is useless. Note that additionally to stability, accuracy and precision are also wanted features for a localisation system, that we won't address here.
\end{rem}
We also briefly explain how our method can be coupled with a paving algorithm to characterize the \emph{stability region} of the localisation system. The latter corresponds to the acceptable set of parameters of the system, in the sense that it remains stable regardless of the parameters picked in that set.

Consider a range-only localisation system using two landmarks $\mathbf{a},\:\mathbf{b}$
as represented on Figure \ref{fig:localisation-using-a}. A static
robot is located at position $\mathbf{m}$, and is able to measure
the exact distance between itself and the landmarks.
\begin{figure}[h]
	\begin{centering}
		\includegraphics{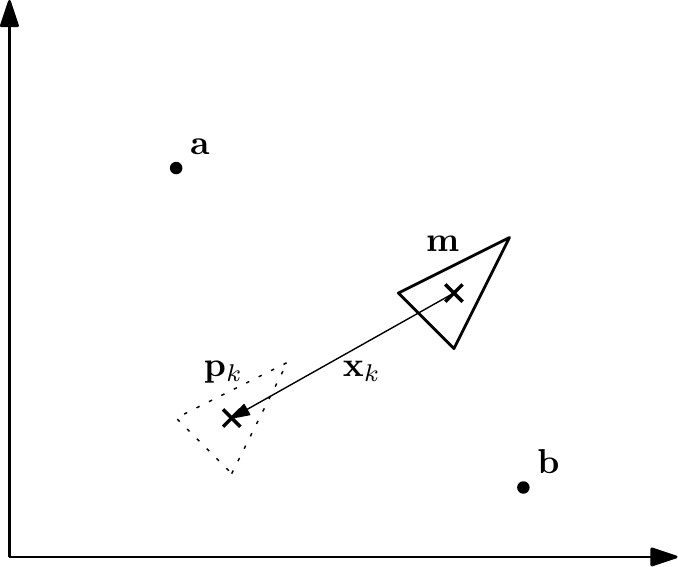}
		\par\end{centering}
	\caption{localisation problem formalisation\label{fig:localisation-using-a}}
\end{figure}

More precisely, we assume that we have the following measurements
\[
\left(\begin{array}{c}
y_{1}\\
y_{2}
\end{array}\right)=\mathbf{h}(\mathbf{m})=\left(\begin{array}{c}
(m_{1}-a_{1})^{2}+(m_{2}-a_{2})^{2}\\
(m_{1}-b_{1})^{2}+(m_{2}-b_{2})^{2}
\end{array}\right).
\]
Assume also that, at iteration $k$, the robot believes to be at position
$\mathbf{p}_{k}$ whereas it actually is at position $\mathbf{m}$.
We define the Newton sequence as
\[
\mathbf{p}_{k+1}=\mathbf{p}_{k}+\mathbf{J}_{\mathbf{h}}^{-1}(\mathbf{p}_{k})\cdot(\mathbf{h}(\mathbf{m})-\mathbf{h}(\mathbf{p}_{k}))
\]
where
\[
\mathbf{J}_{\mathbf{h}}(\mathbf{p})=2\left(\begin{array}{ccc}
p_{1}-a_{1} & \, & p_{2}-a_{2}\\
p_{1}-b_{1} & \, & p_{2}-b_{2}
\end{array}\right)
\]
We have
\[
\mathbf{x}_{k+1}=\underset{\mathbf{f}(\mathbf{x}_{k},\,\mathbf{m})}{\underbrace{\mathbf{x}_{k}+\mathbf{J}_{\mathbf{h}}^{-1}(\mathbf{x}_{k}+\mathbf{m})\cdot(\mathbf{h}(\mathbf{m})-\mathbf{h}(\mathbf{x}_{k}+\mathbf{m}))}},
\]
where $\mathbf{x}_{k}=\mathbf{p}_{k}-\mathbf{m}$ is the localisation
error. We want to show that for $\mathbf{m}$ in a given region of
the plane, the sequence defined by $\mathbf{x}_{k+1}=\mathbf{f}(\mathbf{x}_{k},\mathbf{m})$
is stable,\emph{ i.e.}, converges to $\mathbf{0}$.

\begin{rem}
Of course, this localisation system could be greatly improved. But let us imagine this localisation system has been implemented on a robot and working for years. Its performances are known and trusted, and changing it is not possible or desired. The goal here is to validate the existing system, not to build a new one. The guarantee of stability is a property we can check with our approach.
\end{rem}

Now, let us introduce the concept of \emph{stability region}.
\begin{defn}
	The sequence $\mathbf{x}_{k+1}=\mathbf{f}(\mathbf{x}_{k},\,\mathbf{m})$
	parametrized by $\mathbf{m}\in[\mathbf{m}]$ is \emph{robustly stable} if
	\begin{equation}
	\forall\,\mathbf{m}\in[\mathbf{m}],\:\lim_{k\rightarrow\infty}\mathbf{x}_{k}=\mathbf{0}
	\end{equation}
\end{defn}

\begin{defn}
	We define the\textit{ stability region} as
	\begin{equation}
	\mathbb{M}=\left\{ \mathbf{m}\,|\,\exists\varepsilon>0,\,\forall\,\mathbf{x}_{0}\in[-\varepsilon,\varepsilon]^{\times n},\lim_{k\rightarrow\infty}\|\mathbf{x}_{k}\|=0\right\} 
	\end{equation}
\end{defn}
To calculate an inner approximation of $\mathbb{M}$, we decompose
the $\mathbf{m}$-space into small boxes. Take one of these small boxes
$[\mathbf{m}]$ and a box $[\mathbf{x}_{0}]$ around $\mathbf{0}$.
The box $[\mathbf{x}_{0}]$ should be small, but large compared to
$[\mathbf{m}]$. If there exists $k>0$ such that the system is stable, \textit{i.e.} $[\mathbf{f}_{c}^{k}]([\mathbf{x}_{0}],\,[\mathbf{m}])\subset[\mathbf{x}_{0}]$,
then $[\mathbf{m}]\subset\mathbb{M}$.

Consider a situation where the landmarks are at positions $\mathbf{a}=\left(0,\:0.1\right)^{\mathrm{T}},\,\mathbf{b}=\left(0,\:-0.1\right)^{\mathrm{T}}$.
To characterize the set $\mathbb{M}$, we build a paving of the plane
with small boxes $[\mathbf{m}]$ of width $0.001$. Taking $\varepsilon=\sqrt{w([\mathrm{\mathbf{m}}])}$
and $[\mathbf{x}_{0}]=[-\varepsilon,\varepsilon]^{\times 2}$, we get the
results depicted on Figure \ref{fig:Result-of-the}. The green area
is proved to be inside $\mathbb{M}$. As a consequence, if the robot
is located in the green area and if its initial localisation error
is smaller than $\varepsilon$, then the classical Newton method will
converge towards the actual position of the robot. On the other hand,
we are not able to conclude anything about the stability of the localisation
algorithm in the red area: this could require to take smaller boxes
$[\mathbf{m}]$ in the $\mathbf{m}$-space, a smaller initial condition $[\mathbf{x}_{0}]$, or this might also mean that the localisation system is not stable when initialised inside that specific choice of parameters.
\begin{figure}[h]
	\centering
	\includegraphics[width=0.5\columnwidth]{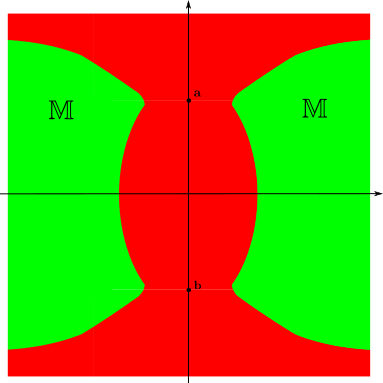}
	\caption{Stability region $\mathbb{M}$ of the localisation system}
	\label{fig:Result-of-the}
\end{figure}

\subsection{Example 4~: Stability of a cyclic trajectory}
This last example illustrates how our method can be applied to a real-life problem~: a robot controlled by a state-machine is deployed in a lake; we will prove that the trajectory of the latter is stable.

\subsubsection{Presentation of the problem}
Let us consider a robot (\textit{e.g.} an Autonomous Underwater Vehicle
(AUV)) cruising in a closed area (\textit{e.g.} a lake) at a constant speed. We consider that the robot is moving in a plane, for the sake of simplicity. The robot is controlled by an automaton and a heading
controller. It also embeds a sensor (\textit{e.g.} a sonar)
to detect if the shore is close by. In this case, an event is triggered
and the automaton changes its states~: a new heading command is
sent to the controller. Consider for instance the following sequence for
the automaton:
\begin{enumerate}
	\item Head East during $\text{25}$ sec.
	\item Head North until reaching the shore
	\item Go South for $7.5$ sec
	\item Head West until reaching the shore
	\item Go to 1.
\end{enumerate}

\subsubsection{Proving stability of the robot's trajectory}
The goal is to prove that the robot's trajectory will converge towards a stable
cycle.

Assume that the border is modelled by the function
\begin{equation}
x_{2}=h(x_{1})=20\left(1-\exp\left(-0.25x_{1}\right)\right)\label{eq:isobath}
\end{equation}

Given the starting position of the robot $\mathbf{x}=(x_{1},x_{2})^{\mathrm{T}}$,
the positions of the robot after each of the four transitions are
respectively given by the following functions~:
\begin{align*}
\mathbf{f}_{1}\left(\mathbf{x}\right)&=\left(x_{1}+25v,\:x_{2}\right)^{\mathrm{T}}\\
\mathbf{f}_{2}\left(\mathbf{x}\right)&=\left(x_{1},\:h\left(x_{1}\right)\right)^{\mathrm{T}}\\
\mathbf{f}_{3}\left(\mathbf{x}\right)&=\left(x_{1},\:x_{2}-7.5v\right)^{\mathrm{T}}\\
\mathbf{f}_{4}\left(\mathbf{x}\right)&=\left(h^{-1}(x_{2}),\:x_{2}\right)^{\mathrm{T}}
\end{align*}
where $v$ denotes the cruising speed of the robot.
Each of these functions can be modified to take model and sensor uncertainties
into account, as we consider that the robot is cruising in dead-reckoning.
Usually, these uncertainties grow with time if the robot does not
perform exteroceptive measurements.

Let us define the cycle function $\mathbf{f}$, given by Equation (\ref{eq:cycle}):
\begin{eqnarray}
\mathbf{f}\left(\mathbf{x}\right) & = & \mathbf{f}_{4}\circ\mathbf{f}_{3}\circ\mathbf{f}_{2}\circ\mathbf{f}_{1}\left(\mathbf{x}\right)\label{eq:cycle}
\end{eqnarray}
Proving stability of the cycle comes down to proving that the discrete-time
system $\mathbf{x}_{k+1}=\mathbf{f}\left(\mathbf{x}_{k}\right)$ is stable.

Figure \ref{fig:Isobath-rebounds} represents the evolution
of $[\mathbf{x}]$ over steps of integration, as well as the intermediate
steps. The initial box is $\left[\mathbf{x}_{0}\right]=\left(\left[1.5,6.5\right],\:\left[9.5,15.5\right]\right)^{\mathrm{T}}$, and we model the robot getting lost by adding an uncertainty $[\mathbf{u}_{i}]=\frac{\left\Vert \mathbf{f}_{i}\left(\mathbf{x}\right)-\mathbf{x}\right\Vert }{v}\cdot\left[-\varepsilon,\varepsilon\right]^{\times2}$ to each intermediate function $\mathbf{f}_{i}$, where $v=1\:m/s$ is the cruising speed of the robot and $\varepsilon=0.05\:m/s$ is the speed at which the robot is getting lost, \textit{i.e.} the inflating rate of the robot's estimated position.
\begin{figure}[h]
	\centering
	\includegraphics[width=0.7\columnwidth]{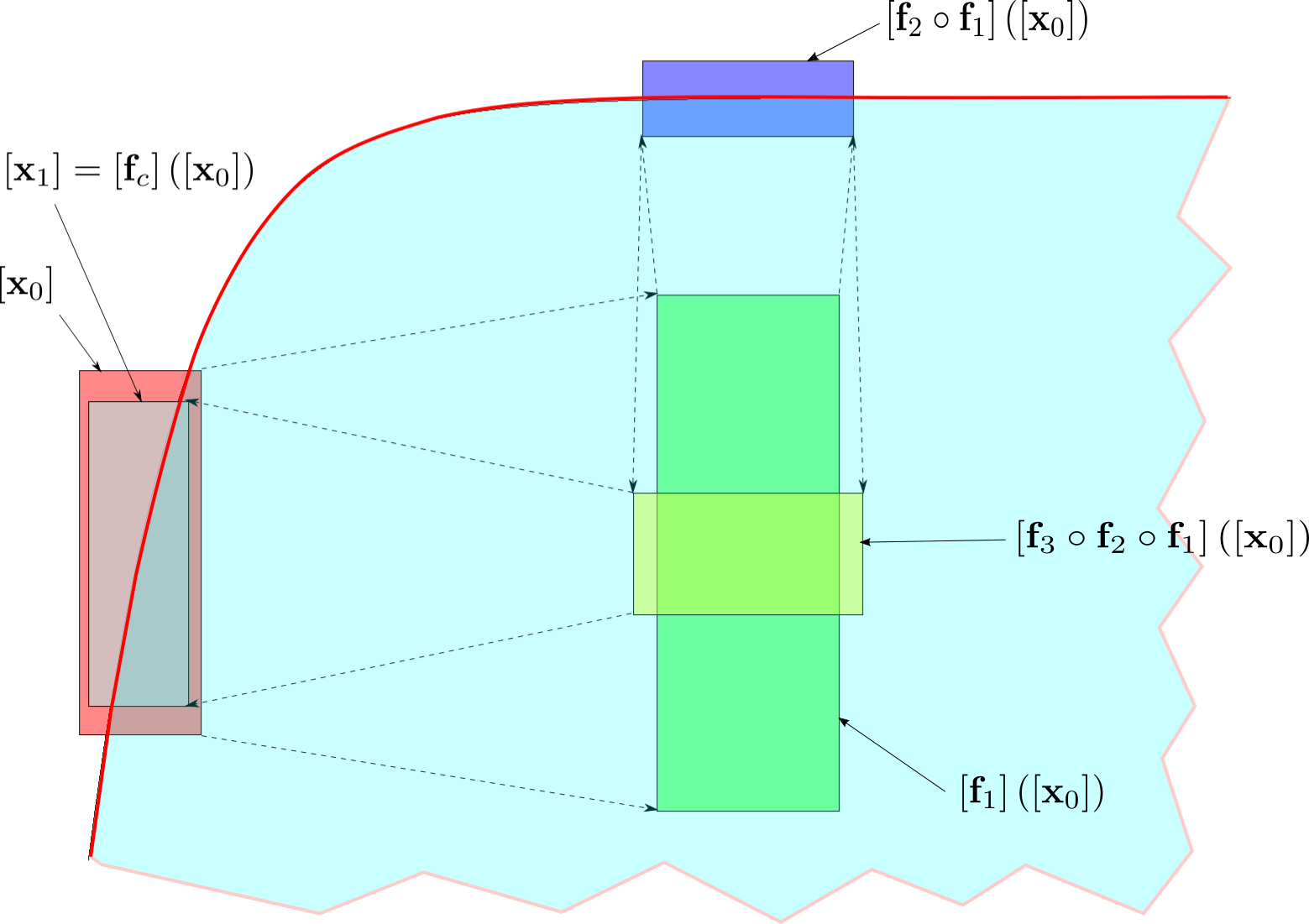}
	\caption{The robot rebounds on the border and follows a cyclic trajectory\label{fig:Isobath-rebounds}}
\end{figure}

Since we have $\left[\mathbf{x}_{1}\right]\in\left[\mathbf{x}_{0}\right]$,
we can conclude that the robot, equipped with the sensors named earlier
and controlled by the given automaton, is going to perform a stable
cycle in the lake.

\section{Conclusion\label{sec:conclusion}}

In this paper, we have proposed a new approach to prove Lyapunov
stability of a non-linear discrete-time system in a given neighbourhood $[\mathbf{x}_0]$, regardless of the system's uncertainties and the floating-point round-off errors. The principle is to
perform a simulation based on the interval centred form, from an
initial box $[\mathbf{x}_{0}]$ and until the current box $[\mathbf{x}_{k}]$
is strictly enclosed in $[\mathbf{x}_{0}]$. From the properties of
the centred form we have proposed, we are able to conclude that a system is stable. The method applies to a large class of
systems and can be used for arbitrary large dimensions, which is not
the case for methods based on linearisation which are very sensitive
to the dimension of the state vector. Additionally, contrary to the existing methods, ours is able to find a neighbourhood of radius $\delta$ inside which the system is stable.

The next step is to extend this approach to deal with continuous-time
systems described by differential equations. This extension will require
the introduction of interval integration methods \cite{chapoutot_rk, Revol05, robloc, wilczak2011}.

The interval community can be decomposed into two worlds~:
\begin{itemize}
	\item On the one hand, \emph{small intervals} \cite{Neumaier90} are used in combination with linearisation methods and
	are able to propagate efficiently small uncertainties such as those
	related to round-off errors. They can solve high dimensional
	problems and do not require any bisections.
	\item On the other hand, \emph{large intervals} are mainly used for global optimization \cite{Hansen},
	solving equations \cite{Kearfott96}, characterizing sets \cite{JaulinBook01}.
	They used together with contractors techniques, constraint propagation and
	bisections algorithms.
\end{itemize}
This work belongs to the first world and makes use of small intervals.
This is why the centred form is very efficient and why we do not
perform any bisection. Now, in practice, the approaches developed by both communities can be used jointly~: to solve real-life problems such as approximating basins of attraction,
we will have to combine the two methodologies appropriately.

\bibliographystyle{eptcs}
\bibliography{interval}

\begin{thebibliography}{10}
\providecommand{\bibitemdeclare}[2]{}
\providecommand{\surnamestart}{}
\providecommand{\surnameend}{}
\providecommand{\urlprefix}{Available at }
\providecommand{\url}[1]{\texttt{#1}}
\providecommand{\href}[2]{\texttt{#2}}
\providecommand{\urlalt}[2]{\href{#1}{#2}}
\providecommand{\doi}[1]{doi:\urlalt{http://dx.doi.org/#1}{#1}}
\providecommand{\bibinfo}[2]{#2}

\bibitemdeclare{inproceedings}{Andrade94}
\bibitem{Andrade94}
\bibinfo{author}{M.~V.~A. \surnamestart Andrade\surnameend},
  \bibinfo{author}{J.~L.~D. \surnamestart Comba\surnameend} \&
  \bibinfo{author}{J.~\surnamestart Stolfi\surnameend} (\bibinfo{year}{1994}):
  \emph{\bibinfo{title}{Affine Arithmetic}}.
\newblock In: {\sl \bibinfo{booktitle}{Interval'94}}, \bibinfo{address}{St
  Petersburg}.

\bibitemdeclare{article}{Asarin07}
\bibitem{Asarin07}
\bibinfo{author}{E.~\surnamestart Asarin\surnameend},
  \bibinfo{author}{T.~\surnamestart Dang\surnameend} \&
  \bibinfo{author}{A.~\surnamestart Girard\surnameend} (\bibinfo{year}{2007}):
  \emph{\bibinfo{title}{Hybridization methods for the analysis of non-linear
  systems}}.
\newblock {\sl \bibinfo{journal}{Acta Informatica}}
  \bibinfo{volume}{7}(\bibinfo{number}{43}), pp. \bibinfo{pages}{451--476},
  \doi{10.1007/s00236-006-0035-7}.

\bibitemdeclare{article}{Chabert09}
\bibitem{Chabert09}
\bibinfo{author}{G.~\surnamestart Chabert\surnameend} \&
  \bibinfo{author}{L.~\surnamestart Jaulin\surnameend} (\bibinfo{year}{2009}):
  \emph{\bibinfo{title}{{A Priori Error Analysis with Intervals}}}.
\newblock {\sl \bibinfo{journal}{SIAM Journal on Scientific Computing}}
  \bibinfo{volume}{31}(\bibinfo{number}{3}), pp. \bibinfo{pages}{2214--2230},
  \doi{10.1137/070696982}.

\bibitemdeclare{inproceedings}{chapoutot_rk}
\bibitem{chapoutot_rk}
\bibinfo{author}{A.~\surnamestart Chapoutot\surnameend},
  \bibinfo{author}{J.~Alexandre~Dit \surnamestart Sandretto\surnameend} \&
  \bibinfo{author}{O.~\surnamestart Mullier\surnameend} (\bibinfo{year}{2015}):
  \emph{\bibinfo{title}{Validated {Explicit} and {Implicit} {Runge}-{Kutta}
  {Methods}}}.
\newblock In: {\sl \bibinfo{booktitle}{Summer Workshop on Interval Methods}}.

\bibitemdeclare{inproceedings}{Combastel05}
\bibitem{Combastel05}
\bibinfo{author}{C.~\surnamestart Combastel\surnameend} (\bibinfo{year}{2005}):
  \emph{\bibinfo{title}{A State Bounding Observer for Uncertain Non-linear
  Continuous-time Systems based on Zonotopes}}.
\newblock In: {\sl \bibinfo{booktitle}{CDC-ECC '05}},
  \doi{10.1109/CDC.2005.1583327}.

\bibitemdeclare{book}{fantoni:02}
\bibitem{fantoni:02}
\bibinfo{author}{I.~\surnamestart Fantoni\surnameend} \&
  \bibinfo{author}{R.~\surnamestart Lozano\surnameend} (\bibinfo{year}{2001}):
  \emph{\bibinfo{title}{Non-linear control for underactuated mechanical
  systems}}.
\newblock \bibinfo{publisher}{Springer-Verlag},
  \doi{10.1007/978-1-4471-0177-2}.

\bibitemdeclare{article}{Frehse:08}
\bibitem{Frehse:08}
\bibinfo{author}{G.~\surnamestart Frehse\surnameend} (\bibinfo{year}{2008}):
  \emph{\bibinfo{title}{PHAVer: Algorithmic Verification of Hybrid Systems}}.
\newblock {\sl \bibinfo{journal}{International Journal on Software Tools for
  Technology Transfer}} \bibinfo{volume}{10}(\bibinfo{number}{3}), pp.
  \bibinfo{pages}{23--48}, \doi{10.1007/978-3-540-31954-2\_17}.

\bibitemdeclare{book}{Hansen}
\bibitem{Hansen}
\bibinfo{author}{E.~R. \surnamestart Hansen\surnameend} (\bibinfo{year}{1992}):
  \emph{\bibinfo{title}{Global Optimization using Interval Analysis}}.
\newblock \bibinfo{publisher}{Marcel Dekker}, \bibinfo{address}{New York, NY}.

\bibitemdeclare{article}{JaulinBurger99}
\bibitem{JaulinBurger99}
\bibinfo{author}{L.~\surnamestart Jaulin\surnameend} \&
  \bibinfo{author}{J.~\surnamestart Burger\surnameend} (\bibinfo{year}{1999}):
  \emph{\bibinfo{title}{Proving stability of uncertain parametric models}}.
\newblock {\sl \bibinfo{journal}{Automatica}}, pp. \bibinfo{pages}{627--632},
  \doi{10.1016/S0005-1098(98)00201-5}.

\bibitemdeclare{book}{JaulinBook01}
\bibitem{JaulinBook01}
\bibinfo{author}{L.~\surnamestart Jaulin\surnameend},
  \bibinfo{author}{M.~\surnamestart Kieffer\surnameend},
  \bibinfo{author}{O.~\surnamestart Didrit\surnameend} \&
  \bibinfo{author}{E.~\surnamestart Walter\surnameend} (\bibinfo{year}{2001}):
  \emph{\bibinfo{title}{Applied {I}nterval {A}nalysis, with {E}xamples in
  {P}arameter and {S}tate {E}stimation, {R}obust {C}ontrol and {R}obotics}}.
\newblock \bibinfo{publisher}{Springer-Verlag}, \bibinfo{address}{London},
  \doi{10.1007/978-1-4471-0249-6}.

\bibitemdeclare{article}{jaulinsliding}
\bibitem{jaulinsliding}
\bibinfo{author}{Luc \surnamestart Jaulin\surnameend} \&
  \bibinfo{author}{Fabrice \surnamestart Bars\surnameend}
  (\bibinfo{year}{2020}): \emph{\bibinfo{title}{Characterizing Sliding Surfaces
  of Cyber-Physical Systems}}.
\newblock {\sl \bibinfo{journal}{Acta Cybernetica}} \bibinfo{volume}{24}, pp.
  \bibinfo{pages}{431--448}, \doi{10.14232/actacyb.24.3.2020.9}.

\bibitemdeclare{book}{Kearfott96}
\bibitem{Kearfott96}
\bibinfo{editor}{R.~B. \surnamestart Kearfott\surnameend} \&
  \bibinfo{editor}{V.~\surnamestart Kreinovich\surnameend}, editors
  (\bibinfo{year}{1996}): \emph{\bibinfo{title}{Applications of Interval
  Computations}}.
\newblock \bibinfo{publisher}{Kluwer}, \bibinfo{address}{Dordrecht, the
  Netherlands}, \doi{10.1007/978-1-4613-3440-8}.

\bibitemdeclare{article}{Lhommeau:Viability:Incinco07}
\bibitem{Lhommeau:Viability:Incinco07}
\bibinfo{author}{M.~\surnamestart Lhommeau\surnameend},
  \bibinfo{author}{L~\surnamestart Jaulin\surnameend} \&
  \bibinfo{author}{L.~\surnamestart Hardouin\surnameend}
  (\bibinfo{year}{2007}): \emph{\bibinfo{title}{Inner and outer approximation
  of capture basins using interval analysis}}.
\newblock {\sl \bibinfo{journal}{ICINCO 2007}}.

\bibitemdeclare{incollection}{Lohner87}
\bibitem{Lohner87}
\bibinfo{author}{R.~\surnamestart Lohner\surnameend} (\bibinfo{year}{1987}):
  \emph{\bibinfo{title}{Enclosing the solutions of ordinary initial and
  boundary value problems}}.
\newblock In \bibinfo{editor}{E.~\surnamestart Kaucher\surnameend},
  \bibinfo{editor}{U.~\surnamestart Kulisch\surnameend} \& \bibinfo{editor}{Ch.
  \surnamestart Ullrich\surnameend}, editors: {\sl \bibinfo{booktitle}{Computer
  Arithmetic: Scientific Computation and Programming Languages}},
  \bibinfo{publisher}{BG Teubner}, \bibinfo{address}{Stuttgart, Germany}, pp.
  \bibinfo{pages}{255--286}.

\bibitemdeclare{inproceedings}{lemezomaze19}
\bibitem{lemezomaze19}
\bibinfo{author}{T.~Le \surnamestart M\'ezo\surnameend},
  \bibinfo{author}{L.~\surnamestart Jaulin\surnameend} \&
  \bibinfo{author}{B.~\surnamestart Zerr\surnameend} (\bibinfo{year}{2019}):
  \emph{\bibinfo{title}{Bracketing backward reach sets of a dynamical system}}.
\newblock In: {\sl \bibinfo{booktitle}{International Journal of Control}},
  \doi{10.1080/00207179.2019.1643910}.

\bibitemdeclare{book}{Moore66}
\bibitem{Moore66}
\bibinfo{author}{R.~E. \surnamestart Moore\surnameend} (\bibinfo{year}{1966}):
  \emph{\bibinfo{title}{Interval Analysis}}.
\newblock \bibinfo{publisher}{Prentice-Hall}, \bibinfo{address}{Englewood
  Cliffs, NJ}.

\bibitemdeclare{book}{Moore79}
\bibitem{Moore79}
\bibinfo{author}{R.~E. \surnamestart Moore\surnameend} (\bibinfo{year}{1979}):
  \emph{\bibinfo{title}{Methods and {A}pplications of {I}nterval {A}nalysis}}.
\newblock \bibinfo{publisher}{SIAM}, \bibinfo{address}{Philadelphia, PA},
  \doi{10.1137/1.9781611970906}.

\bibitemdeclare{book}{moore2009introduction}
\bibitem{moore2009introduction}
\bibinfo{author}{Ramon~E \surnamestart Moore\surnameend},
  \bibinfo{author}{R~Baker \surnamestart Kearfott\surnameend} \&
  \bibinfo{author}{Michael~J \surnamestart Cloud\surnameend}
  (\bibinfo{year}{2009}): \emph{\bibinfo{title}{Introduction to interval
  analysis}}.
\newblock \bibinfo{publisher}{SIAM}, \doi{10.1137/1.9780898717716}.

\bibitemdeclare{book}{Neumaier90}
\bibitem{Neumaier90}
\bibinfo{author}{A.~\surnamestart Neumaier\surnameend} (\bibinfo{year}{1991}):
  \emph{\bibinfo{title}{Interval Methods for Systems of Equations}}.
\newblock \bibinfo{publisher}{Cambridge University Press},
  \bibinfo{address}{Cambridge, UK}, \doi{10.1017/CBO9780511526473}.

\bibitemdeclare{article}{Ramdani:Nedialkov11}
\bibitem{Ramdani:Nedialkov11}
\bibinfo{author}{N.~\surnamestart Ramdani\surnameend} \&
  \bibinfo{author}{N.~\surnamestart Nedialkov\surnameend}
  (\bibinfo{year}{2011}): \emph{\bibinfo{title}{{Computing Reachable Sets for
  Uncertain Nonlinear Hybrid Systems using Interval Constraint Propagation
  Techniques}}}.
\newblock {\sl \bibinfo{journal}{Nonlinear Analysis: Hybrid Systems}}
  \bibinfo{volume}{5}(\bibinfo{number}{2}), pp. \bibinfo{pages}{149--162},
  \doi{10.1016/j.nahs.2010.05.010}.

\bibitemdeclare{article}{ratschan_stab}
\bibitem{ratschan_stab}
\bibinfo{author}{S.~\surnamestart Ratschan\surnameend} \&
  \bibinfo{author}{Z.~\surnamestart She\surnameend} (\bibinfo{year}{2010}):
  \emph{\bibinfo{title}{Providing a {Basin} of {Attraction} to a {Target}
  {Region} of {Polynomial} {Systems} by {Computation} of {Lyapunov}-{Like}
  {Functions}}}.
\newblock {\sl \bibinfo{journal}{SIAM Journal on Control and Optimization}}
  \bibinfo{volume}{48}(\bibinfo{number}{7}), pp. \bibinfo{pages}{4377--4394},
  \doi{10.1137/090749955}.

\bibitemdeclare{article}{Revol05}
\bibitem{Revol05}
\bibinfo{author}{N.~\surnamestart Revol\surnameend},
  \bibinfo{author}{K.~\surnamestart Makino\surnameend} \&
  \bibinfo{author}{M.~\surnamestart Berz\surnameend} (\bibinfo{year}{2005}):
  \emph{\bibinfo{title}{Taylor models and floating-point arithmetic: proof that
  arithmetic operations are validated in {COSY}}}.
\newblock {\sl \bibinfo{journal}{Journal of Logic and Algebraic Programming}}
  \bibinfo{volume}{64}, pp. \bibinfo{pages}{135--154},
  \doi{10.1016/j.jlap.2004.07.008}.

\bibitemdeclare{article}{Rohn:stab:mat:96}
\bibitem{Rohn:stab:mat:96}
\bibinfo{author}{J.~\surnamestart Rohn\surnameend} (\bibinfo{year}{1996}):
  \emph{\bibinfo{title}{An algorithm for checking stability of symmetric
  interval matrices}}.
\newblock {\sl \bibinfo{journal}{IEEE Trans. Autom. Control}}
  \bibinfo{volume}{41}(\bibinfo{number}{1}), pp. \bibinfo{pages}{133--136},
  \doi{10.1109/9.481618}.

\bibitemdeclare{book}{robloc}
\bibitem{robloc}
\bibinfo{author}{S.~\surnamestart Rohou\surnameend},
  \bibinfo{author}{L.~\surnamestart Jaulin\surnameend},
  \bibinfo{author}{L.~\surnamestart Mihaylova\surnameend},
  \bibinfo{author}{F.~\surnamestart {Le Bars}\surnameend} \&
  \bibinfo{author}{S.~\surnamestart Veres\surnameend} (\bibinfo{year}{2019}):
  \emph{\bibinfo{title}{Reliable robot localization}}.
\newblock \bibinfo{publisher}{ISTE Group}, \doi{10.1002/9781119680970}.

\bibitemdeclare{incollection}{Rump00}
\bibitem{Rump00}
\bibinfo{author}{S.~M. \surnamestart Rump\surnameend} (\bibinfo{year}{2001}):
  \emph{\bibinfo{title}{{INTLAB} - {INTerval} {LABoratory}}}.
\newblock In \bibinfo{editor}{J.~\surnamestart Grabmeier\surnameend},
  \bibinfo{editor}{E.~\surnamestart Kaltofen\surnameend} \&
  \bibinfo{editor}{V.~\surnamestart Weispfennig\surnameend}, editors: {\sl
  \bibinfo{booktitle}{Handbook of Computer Algebra: Foundations, Applications,
  Systems}}, \bibinfo{publisher}{Springer-Verlag},
  \bibinfo{address}{Heidelberg, Germany}.

\bibitemdeclare{inproceedings}{SaintPierre02}
\bibitem{SaintPierre02}
\bibinfo{author}{P.~\surnamestart Saint-Pierre\surnameend}
  (\bibinfo{year}{2002}): \emph{\bibinfo{title}{Hybrid kernels and capture
  basins for impulse constrained systems}}.
\newblock In \bibinfo{editor}{C.J. \surnamestart Tomlin\surnameend} \&
  \bibinfo{editor}{M.R. \surnamestart Greenstreet\surnameend}, editors: {\sl
  \bibinfo{booktitle}{in Hybrid Systems: Computation and Control}},
  \bibinfo{volume}{2289}, \bibinfo{publisher}{Springer-Verlag}, pp.
  \bibinfo{pages}{378--392}, \doi{10.1007/3-540-45873-5\_30}.

\bibitemdeclare{book}{slotine91}
\bibitem{slotine91}
\bibinfo{author}{J.J. \surnamestart Slotine\surnameend} \&
  \bibinfo{author}{W.~\surnamestart Li\surnameend} (\bibinfo{year}{1991}):
  \emph{\bibinfo{title}{Applied nonlinear control}}.
\newblock \bibinfo{publisher}{Prentice Hall}, \bibinfo{address}{Englewood
  Cliffs (N.J.)}.
\newblock \urlprefix\url{http://opac.inria.fr/record=b1132812}.

\bibitemdeclare{inproceedings}{taha:15:acumen}
\bibitem{taha:15:acumen}
\bibinfo{author}{W.~\surnamestart Taha\surnameend} \&
  \bibinfo{author}{A.~\surnamestart Duracz\surnameend}:
  \emph{\bibinfo{title}{Acumen: An Open-source Testbed for Cyber-Physical
  Systems Research}}.
\newblock In: {\sl \bibinfo{booktitle}{CYCLONE'15}},
  \doi{10.1007/978-3-319-47063-4\_11}.

\bibitemdeclare{phdthesis}{JianWan07}
\bibitem{JianWan07}
\bibinfo{author}{Jian \surnamestart Wan\surnameend} (\bibinfo{year}{2007}):
  \emph{\bibinfo{title}{Computationally reliable approaches of contractive
  model predictive control for discrete-time systems}}.
\newblock \bibinfo{type}{{PhD} dissertation}, \bibinfo{school}{Universitat de
  Girona}, \bibinfo{address}{Girona, Spain}.

\bibitemdeclare{article}{wilczak2011}
\bibitem{wilczak2011}
\bibinfo{author}{D.~\surnamestart Wilczak\surnameend} \&
  \bibinfo{author}{P.~\surnamestart Zgliczynski\surnameend}
  (\bibinfo{year}{2011}): \emph{\bibinfo{title}{{C}r-{L}ohner algorithm}}.
\newblock {\sl \bibinfo{journal}{Schedae Informaticae}} \bibinfo{volume}{20},
  pp. \bibinfo{pages}{9--46}, \doi{10.4467/20838476SI.11.001.0287}.

\end{thebibliography}
\end{document}